\documentclass[12pt,draftcls,onecolumn,journal]{IEEEtran} 

\usepackage{amsfonts}
\usepackage{times}
\usepackage{latexsym}
\usepackage{amssymb}
\usepackage{amsmath}
\usepackage{cite}
\usepackage{verbatim}
\def\bb0{{\mathbb{0}}}

\def\bb{{\mathbf{b}}}

\def\bp{{\mathbf{p}}}

\def\b0{{\mathbf{0}}}

\def\bA{{\mathbf{A}}}

\def\b1{{\mathbf{1}}}

\def\sf0{{\mathsf{0}}}

\def\nn{\nonumber}

\usepackage{amssymb}
\usepackage{amsxtra}
\usepackage{amsmath}
\usepackage{amsthm}
\usepackage{mathtools}
\DeclarePairedDelimiter{\ceil}{\lceil}{\rceil}
\DeclarePairedDelimiter{\floor}{\lfloor}{\rfloor}

\newtheorem{theorem}{Theorem}
\newtheorem{prop}[theorem]{Proposition}
\newtheorem{lemma}[theorem]{Lemma}

\usepackage{amsmath,amssymb}
\usepackage{graphicx}
\usepackage{color}
\usepackage{cite}
\usepackage{multirow,tabularx}
\usepackage{ifthen}
\usepackage{times}
\usepackage{graphicx}
\usepackage{amsbsy} 
\usepackage{epsfig}
\usepackage{color}

\usepackage{stackrel}

\usepackage{amsfonts}
\usepackage{stfloats}
\usepackage[all]{xy}

\newcommand{\e}{\epsilon}
\newcommand{\M}{\mathcal{M}}
\newcommand{\A}{\mathcal{A}}
\renewcommand{\O}{\mathcal{O}}

\newcommand{\tA}{\tilde{\mathcal{A}}}
\newcommand{\W}{\mathcal{W}}
\newcommand{\w}{\mathbf{w}}
\newcommand{\off}{\text{off}}

\newcommand{\avg}{\text{avg}}
\newcommand{\pr}{\text{Pr}}
\newcommand{\worst}{\text{worst}}
\newcommand{\MWM}{\text{MWM}}
\newcommand{\mw}{\text{mw}}
\newcommand\undermat[2]{%
  \makebox[0pt][l]{$\smash{\underbrace{\phantom{%
    \begin{matrix}#2\end{matrix}}}_{\text{$#1$}}}$}#2}

\begin{document}
\title{Online Algorithms for Basestation Allocation}
\author{Andrew Thangaraj and Rahul Vaze
\thanks{A. Thangaraj is with the Department of Electrical Engineering, Indian Institute of Technology Madras, email: andrew@ee.iitm.ac.in}
\thanks{R. Vaze is with the School of Technology and Computer Science, Tata Institute of Fundamental Research Mumbai, email: vaze@tcs.tifr.res.in}}
\maketitle
\begin{abstract} Design of {\it online algorithms} for assigning mobile users to basestations is considered with the objective of maximizing the sum-rate, when all users associated to any one basestation equally  share each basestation's resources. Each user on its arrival reveals the rates it can obtain if connected to each of the basestations, and the problem is to assign each user to any one basestation irrevocably so that the sum-rate is maximized at the end of all user arrivals, without knowing the future user arrival or rate information or its statistics at each user arrival. 
Online algorithms with constant factor loss in comparison to offline algorithms (that know both the user arrival and user rates profile in advance) are derived. The proposed online algorithms are motivated from the famous online $k$-secretary problem and online maximum weight matching problem.

\end{abstract}
\section{Introduction}
\label{sec:introduction}
Associating mobile users to basestations under different utility models is a classical problem that has been well studied in literature. Many different utility models have been considered in prior work including load balancing \cite{Viswanathan2003LoadBalancing, Chen2007ProcSharing, DeVeciana2009LoadBalancing, Altman2011LoadBalancing, DeVeciana2012userAssociation}, cell breathing \cite{Bejerano2009CellBreathing}, call admission \cite{Wu2000MACA}, fairness \cite{Bejerano2004Fairness}. Another class of problems that resembles basestation association is that of sub-carrier/power allocation in a multi-carrier system either jointly with base-station allocation \cite{Subramanian2009userAssociation, Subramanian2009auserAssociation} or without it \cite{Kivanc2000SpectrumPowerAllocation, Kim2005SpectrumPowerAllocation, Yates2009SpectrumAllocation}.

Almost all of the prior work on basestation association or sub-carrier allocation either assumes that exact information about all the users is available (number of users and their channel gains) or their statistics is known and solves a joint optimization problem. In practice, however, loads in cellular networks are dynamic, i.e. users keep coming and going out of the system, and decisions made at any fixed time cease to be optimal in a short time window. More importantly, most prior work assumes an homogenous model, where the rates of each user achievable from different basestations are independent and identically distributed. This is a serious limitation in practice, since there is correlation in rates obtained by users that are physically close to each other, and there is correlation between rate obtained by a user from different basestations because of propagation environment, e.g. if the user is within a building or has severe shadowing losses or has highly directional antennas.

To address this issue, we take a different viewpoint and design {\it online} algorithms for the basestation association problem that do not assume any information about the statistics of the user's profile, i.e. the rates obtained by different users from different basestations can be arbitrarily correlated. In an online algorithm, each user arrives sequentially and is allocated to one of the possible basestations and the association once made cannot be revoked. We assume a {\it closed} system following \cite{Ganesh2010LoadBalancing}, where a total of $N$ users arrive one at a time and stay in the system for the full time duration of interest.  We consider the downlink basestation association problem, where each user on its arrival reveals the rates (computed by basestations using SINR) it can obtain on each of the basestations, and we are interested in finding online algorithms for associating each mobile user to one of the basestations for maximizing the sum-rate at the end of all $N$ user arrivals. We assume that all the users associated with any one basestation equally share the time/bandwidth. The uplink problem is more challenging because of the need for power control and presently out of scope of this paper.

We note that by considering a fixed time duration of interest, reasonable estimates of $N$ can be obtained. Of course, in a practical system, users also leave the system once their service is over; however, by constraing the time duration, we can reasonably handle this restriction. Allowing general {\it user exit} functionality is currently out of scope of this paper.

The performance of an online algorithm is characterized by its \emph{competitive ratio}, which is defined as the ratio of the utility of the offline algorithm (that knows the complete future inputs) to that of the online algorithm \cite{BorodinOnlineBook}. The problem is then to find online algorithms with smallest possible competitive ratio. Note that online algorithms are robust since no assumptions are made about the actual system parameters or their statistics.

Online algorithms have been designed for many related problems in literature, 
such as the following: (a) Load balancing \cite{Azar1998LoadBalancing}, where given a set of servers the problem is to assign each incoming user to minimize the maximum load on any server, (b) Load balancing with deadlines \cite{Moharir2013}, 
where given a set of servers the problem is to assign each incoming user to minimize the maximum load on any server such that each job is completed within the declared deadline, 
(c) Maximum weight matching \cite{Karp1990Matching, Korula2009}, where the problem is to find  a bipartite matching, where the sum of the values of the edges in the matching are maximized, (d) Picking best subset of fixed cardinality \cite{LleinbergK-Secretary2005, babaioff2007k-Secretary} (called the $k$-secretary problem), 
where a fixed number of secretaries are going to arrive in a random order and the problem is to pick the $k$-best secretaries, and 
(e) Multi-partitioning \cite{Nemhauser1978, Chekuri2011MultiPartition}, where given a set $\bA$, the problem is to partition it into $k$ subsets $A_1, \dots, A_k, A_i \cap A_j = \phi, \cup_{i=1}^k A_i = \bA$ such that $\sum_{i=1}^k f_i(A_i)$ is maximized. 
Analytical results on the competitive ratio for the multi-partitioning problem is known only when the functions $f_i$ are sub-modular \cite{Nemhauser1978}.

The basestation association problem described above 
can be thought of as the multi-partitioning problem, 
where the set of users is partitioned into number of subsets equal to the number of basestations, such that the sum-rate is maximized. 
However, since time-sharing utility function  is not sub-modular, prior results on 
competitive ratio for multi-partitioning \cite{Nemhauser1978, Chekuri2011MultiPartition} are not applicable.

Our contributions are as follows:
\begin{itemize}
\item Under the worst case input, we show that for the case when each user has identical rates to all basestations, the competitive ratio is of the order $n/m$, where $n$ is the number of users and $m$ is the number of basestations. For the general case of each user having different rates to each basestation the competitive ratio further worsens to order $n$. The results reveal that online algorithms are too pessimistic in the worst case.

\item We next consider a more reasonable randomized input model, where the user rates to each basestation can still be worst case (adversarial); however, the order in which users arrive is randomized. 
We show that for the case when each user has identical rates to all basestations, the competitive ratio is close to $2$. For the general case of each user having different rates to each basestation, the competitive ratio is shown to be at most equal to $8$. This is a substantial improvement over the worst case input, since with the randomized input the competitive ratio is independent of the parameters of the problem. Moreover, these results show that not only online algorithms are robust, being independent of any change in system parameters and statistics, they incur only constant factor loss compared to offline algorithms that know everything about future arrivals in advance.
We also show that the often-used algorithm of attaching each user to the strongest basestation is far from optimal and its competitive ratio scales as $n$ even for the randomized input model.

\item We also consider the case of reassignment, where at any time slot one of the previously assigned user can be reassigned. We show a remarkable improvement in terms of competitive ratio performance. For the case when each user has identical rates to all basestations, the competitive ratio is shown to be equal to $1$, i.e. offline and online algorithm perform the same. 
For the general case of each user having different rates to each basestation, the competitive ratio is shown to be $2$, a four-fold increase over the no reassignment case. 
\end{itemize}

\section{Basestation Allocation Problem: Model and Preliminaries}
\label{sec:basest-alloc-probl}
Consider the scenario where $n$ mobile users arrive one at a time into a geographical area with $m$ basestations with $m<n$. User $i$ is said to have a non-negative weight $w_{ij}$ to basestation $j$, and this weight typically corresponds to the user's achievable rate/capacity. 
The weights $w_{ij}$ are collected into an $n\times m$ matrix $W$. The users, basestations, and the weights are conveniently captured in a weighted complete bipartite graph $G=(V_1\cup V_2,V_1\times V_2)$, where $V_1=\{1,2,\ldots,n\}$ is the set of users, $V_2=\{1,2,\ldots,m\}$ is the set of basestations, and the weight of the edge $(i,j)$ is $w_{ij}$. 

In the downlink basestation allocation problem, each user is to be allotted to exactly one of the $m$ basestations. An arbitrary allocation is, therefore, specified by the family of sets $\M=\{M_j: 1\le j\le m\}$, where $M_j$ denotes the set of users allotted to basestation $j$. Note that the sets $M_j$ partition $V_1$. We consider the allocation problem, where the goal is to find an allocation $\M$ that maximizes a certain objective function or utility denoted $R(\M,W)$, which is a function of the weight matrix as well as the allocation. The allocation made by an algorithm $\A$ for a weight matrix $W$ will be denoted $\M_\A(W)$. 

\subsection{Offline versus online allocation}
In an {\it offline} allocation problem, the entire weight matrix $W$ is available ahead of time. Therefore, an offline algorithm attempts to find the optimal allocation by solving the following problem:
\begin{equation}
  \label{eq:8}
  \M^*_\off(W)=\arg\max_{\M}R(\M,W).
\end{equation}
Let us denote the optimal offline utility by $R^*_\off=R(\M^*_\off(W),W)$.

In an {\it online} allocation problem, the users arrive in an arbitrary order, and weights of user $i$ are revealed only upon arrival. i.e. the $i$-th row of $W$ is revealed when user $i$ arrives. Most importantly, user $i$ needs to be allotted to a basestation upon arrival, and this allocation cannot be altered subsequently. An important distinction between online and offline algorithms is that an online algorithm's utility and allocation vary with a permutation of the rows of $W$, i.e. users order of arrival, while optimal offline algorithms are invariant to row permutations. So, in our study, a row permuted version of $W$ is considered distinct from $W$.

It is readily seen that, for a given weight matrix $W$, the utility of any online algorithm $\A$ will satisfy $R(\M_\A(W),W)\le R^*_\off$. Competitive ratio, denoted $\eta$, is a figure of merit used to characterize and compare online and offline algorithms \cite{BorodinOnlineBook}. For a given online algorithm $\A$ and weight matrix $W$, the competitive ratio is defined as
\begin{equation}
  \label{eq:12}
  \eta_W(\A) = \dfrac{R(\M^*_\off(W),W)}{R(\M_\A(W),W)}.
\end{equation}
The worst-case competitive ratio of a specific online algorithm $\A$ is defined as
\begin{equation}
  \label{eq:9}
  \eta_\worst(\A) = \max_W \eta_W(\A) = \max_W \dfrac{R(\M^*_\off(W),W)}{R(\M_\A(W),W)}.
\end{equation}
The average-case competitive ratio of $\A$ over a distribution on $W$ is defined as
\begin{equation}
  \label{eq:10}
  \eta_\avg(\A) = \sum_W \dfrac{R(\M^*_\off(W),W)}{R(\M_\A(W),W)}\pr(W),
\end{equation}
where $\pr(W)$ denotes the probability of occurance of the weight matrix $W$. The best possible worst-case competitive ratio that can be achieved by any online algoritm is denoted $\eta^*_\worst$ and is defined as
\begin{equation}
  \label{eq:11}
  \eta^*_\worst=\min_\A \eta_\worst(\A)=\min_\A \max_W \eta_W(\A).
\end{equation}
The following lower and upper bounds on $\eta^*_\worst$ are readily established:
\begin{equation}
  \label{eq:13}
  \min_\A \max_{W\in \W} \eta_W(\A) \le \eta^*_\worst \le \eta_\worst(\A_0),
\end{equation}
where $\W$ is a finite set of weight matrices and $\A_0$ is any online algorithm. 

Ideally, we would like to have online algorithms with a competitive ratio as close to 1 as possible. If that is not possible, a competitive ratio that is constant with increasing parameters of the problem is desirable. An undesirable scenario is when the competitive ratio increases without bound, in which case the online algorithm becomes a very poor replacement for the offline one.
\subsection{Utility functions}
Most common and popular utility function in wireless cellular networks is the time-sharing utility, where all users associated to any one basestation equally share its resources. Time sharing utility achieves sum-rate of 
\begin{equation}
  \label{eq:7}
  TS(\M,W)=\sum_{j=1}^m\left(\sum_{i\in M_j}\dfrac{w_{ij}}{d_j}\right),
\end{equation}
where $w_{ij}$ is the rate/capacity of the channel of user $i$ to basestation $j$, and $d_j=|M_j|$ is the degree of basestation $j$ (number of users allocated to basestation $j$) in the allocation $\M$, with each basestation equally time-sharing the allotted users. As shown later, the $1/d$-factor in the time-sharing utility plays a crucial role in determining the competitive ratio.

Note that since we are considering the downlink, achievable rate of any user $i$ associated to basestation $j$ does not depend on the users associated to basestation $k\ne j$.  Essentially, 
we are assuming that all the interfering basestations create identical interference profile at any user  irrespective of the number of users associated with them. Thus, at each user arrival, it is sufficient to know the corresponding user's achievable rates from different basestations (with time sharing), and not how it affects other already assigned users.



\section{Bounds on Worst Case Competitive Ratio}
\label{sec:bounds-worst-case}
For studying the worst case competitive ratio of online algorithms for basestation allocation, we begin by considering different scenarios that result in different forms for the weight matrix $W$. For each case, we derive lower and upper bounds on the worst-case competitive ratio of online algorithms. Unless mentioned otherwise, the utility function will be assumed to be the time-sharing utility of (\ref{eq:7}).

\subsection{Case I: Identical users}
In this case, each basestation has the same weight to all users. Hence, the matrix $W$ is such that all rows are repetitions of the first row, i.e. $w_{ik}=w_{jk}$, $1\le k\le m$ for all $1\le i,j\le n$. So, clearly, the order of arrival of the users is irrelevant, and the information available with the online algorithm is the exact same as that of the offline algorithm. When the first user arrives, the first row of $W$ is revealed to an online algorithm, and this reveals the entire matrix $W$. Therefore, a competitive ratio of 1 can be achieved in this case. In fact, a greedy algorithm that maximizes the utility upon the arrival of each user will be optimal for any utility function. 

\subsection{Case II: Identical basestations}
In this case, each user has the same weight to all basestations making the rows of $W$ constant. Hence,
the matrix $W$ is such that all columns are repetitions of the first column, i.e. $w_{ik}=w_{il}$, $1\le i\le n$ for all $1\le k,l\le m$. The order of arrival is now significant and the offline algorithm has significantly more information than an online algorithm in the worst case. We now study this case closely and show that $\eta^*_\worst$ scales as $n/m$ in the worst case.

\subsubsection{Optimal offline algorithm}\label{subsec:optoffline} We begin by deriving the optimal offline algorithm for the time-sharing utility. An offline algorithm has non-causal knowledge of the weight of user $i$, denoted $w_i$, which is the constant value in row $i$ of $W$.
\begin{prop}
  Let the users be numbered such that $w_1\ge w_2\ge\cdots\ge w_n$. Then, the optimal offline time-sharing utility is
\begin{equation}
  \label{eq:2}
  TS({\cal M}^*_\off(W), W)=w_1+w_2+\cdots+w_{m-1}+\dfrac{1}{n-(m-1)}(w_m+w_{m+1}+\cdots+w_n),
\end{equation}
and an optimal allocation is $M_j=\{j\}$, $1\le j\le m-1$, and $M_m=\{m,m+1,\ldots,n\}$, i.e. the top $m-1$ users are individually allotted to a basestation each, and the remaining $n-m+1$ users share one basestation.
\label{thm:optim-offl-algor}
\end{prop}
\begin{proof}
To prove (\ref{eq:2}), we need the following lemma.
\begin{lemma}
  Let $w_1\ge w_2\ge\cdots\ge w_d$. Then, for any $k$-element subset $S$ of $\{1,2,\ldots,d\}$ with $1\le k\le d-1$, we have
  \begin{equation}
    \label{eq:3}
    w_1+\dfrac{1}{d-1}(w_2+w_3+\cdots+w_d)\ge \dfrac{1}{k+1}(w_1+\sum_{i\in S}w_i)+\dfrac{1}{d-k-1}\sum_{i\notin S\cup \{1\}}w_i.
  \end{equation}
\label{lem:1}
\end{lemma}
\begin{proof}
Simplifying (\ref{eq:3}), we need
$$\dfrac{k}{k+1}w_1\ge \left(\dfrac{1}{k+1}-\dfrac{1}{d-1}\right)\sum_{i\in S}w_i+\left(\dfrac{1}{d-k-1}-\dfrac{1}{d-1}\right)\sum_{i\notin S\cup \{1\}}w_i,$$
which reduces to the requirement
$$k(d-1)(d-k-1)w_1\ge (d-k-1)(d-k-2)\sum_{i\in S}w_i+k(k+1)\sum_{i\notin S\cup \{1\}}w_i.$$
Now, using $\sum_{i\in S}w_i\le kw_1$ and $\sum_{i\notin S\cup \{1\}}w_i\le (d-k-1)w_1$, we see that the above requirement is satisfied.
\end{proof}
In words, Lemma \ref{lem:1} says that, for $m=2$ basestations and $n=d$ users, the maximum throughput is obtained by allotting the user with the largest weight to one basestation, and letting all other users share the other basestation. 

We now consider $m$ basestations and $n$ users. To prove Proposition \ref{thm:optim-offl-algor}, we show that for an arbitrary allocation $\M$,  $TS(\M,W)\le TS({\cal M}^*_\off(W), W)$ by repeated application of Lemma \ref{lem:1}. Let $TS(\M,W)$ allocate $k_j$ users to basestation $j$, such that $\sum_{j=1}^m k_j=n$. Then 

\begin{eqnarray*}
TS(\M,W) &=& \sum_{j=1}^m \frac{1}{k_j} \sum_{i \in M_j} w_i, \\
&=& \frac{1}{k_1} \sum_{i \in M_1} w_i + \frac{1}{k_2} \sum_{i \in M_2} w_i + \sum_{j=3}^m \frac{1}{k_j} \sum_{i \in M_j} w_i\\
\end{eqnarray*}
Without loss of generality assume that $w'_1 \ge w'_2$, where $w'_j=\max_{i\in M_j}w_i$.
Then applying Lemma \ref{lem:1} for basestation $1$ and $2$ by moving all users from basestation $1$ to $2$ except the one with weight $w'_1$, we have 
\begin{eqnarray*}
TS(\M,W) &\le &  w'_1 + \frac{1}{k_2+k_1-1} \sum_{i \in M_2, |M_2| = k_2+k_1-1} w_i + \sum_{j=3}^m \frac{1}{k_j} \sum_{i \in M_j} w_i.
\end{eqnarray*}
Repeating this argument successively for basestation $i$ and $i+1$ for $i=2,\dots,n-1$, we get the claim that $TS(\M,W)\le TS({\cal M}^*_\off(W), W)$.

\end{proof}
\subsubsection{Upper bound on the competitive ratio}
We consider a round robin online algorithm $RR$ that allots user $i$ to basestation $1+(i\mod m)$. After $n$ users have arrived, each basestation has either $\floor*{\frac{n}{m}}$ or $\ceil*{\frac{n}{m}}$ users associated to it. $RR$ algorithm achieves a time-sharing utility
\begin{equation}
  \label{eq:5}
  TS(\M_{RR}(W),W)\ge \dfrac{w_1+w_2+\cdots+w_m}{\ceil*{\frac{n}{m}}},
\end{equation}
where, once again, the weights are such that $w_1\ge w_2\ge\cdots\ge w_n$. Since 
$TS^*_{\off}\le w_1+w_2+\cdots+w_m$,
we see that the round robin online algorithm has competitive ratio
$\eta_\worst(RR)\le \ceil*{\frac{n}{m}}$.
Setting $\A_0=RR$ in (\ref{eq:13}), we see that 
\begin{equation}
  \label{eq:14}
\eta^*_\worst \le \ceil*{\frac{n}{m}}.  
\end{equation}
\subsubsection{Lower bound on the competitive ratio}\label{subsec:LBId}
To lower bound the competitive ratio, we use (\ref{eq:13}) with a suitable set of weight matrices $\W$. Denoting the constant value in row $i$ of $W$ as $w_i$, a weight matrix $W$ is specified by the vector $[w_1,w_2,\ldots,w_n]$. We let
\begin{equation}
  \label{eq:15}
\W=\{W_l=[\beta,\beta^2,\ldots,\beta^l,0,\ldots,0]:\ 1\le l\le n\},
\end{equation}
which contains $n$ weight matrices with $\beta>1$ being a constant. From Proposition \ref{thm:optim-offl-algor}, by allocating only the $l^{th}$ user in $W_l$ to one of the basestations, we have 
\begin{equation}
  \label{eq:1}
  TS(\M^*_\off(W_l),W_l)\ge \beta^l,\ 1\le l\le n.
\end{equation}
Now, let us consider an arbitrary online algorithm $\A$. Suppose that, for $W=W_n$, $\A$ allocates user $i$ to basestation $k_i$, and let $d_{\max}(i)$ be the maximum degree of any basestation immediately after user $i$ has been allocated. Since there are $n$ users and $m$ basestations, there exists $i=i^*$ such that $d_{\max}(i^*)\ge n/m$, by a pigeon-hole argument. Now, since $\A$ is online and $W_l$ matches with $W_n$ in the first $l$ rows, for $W=W_l$, $\A$ allocates user $i$ to basestation $k_i$ for $1\le i\le l$. Specifically, for input $W=W_{i^*}$, $\A$ allocates user $i^*$ to a basestation with degree at least $n/m$. Therefore, 
\begin{equation}
  \label{eq:16}
  TS(\M_A(W_{i^*}),W_{i^*})\le \dfrac{\beta^{i^*}}{(n/m)}+m\beta^{i^*-1},
\end{equation}
since the contribution of each basestation is at most $\beta^{i^*-1}$ for all users other than user $i^*$. Using (\ref{eq:1}) and (\ref{eq:16}), we get
\begin{equation}
  \label{eq:17}
  \eta_{W_{i^*}}(\A)= \dfrac{TS(\M^*_\off(W_{i^*}),W_{i^*})}{TS(\M_A(W_{i^*}),W_{i^*})}\to\dfrac{n}{m}, \text{ as }\beta\to\infty.
\end{equation}
Therefore, for every online algorithm $\A$, $\max_{W\in\W}\eta_W(\A)\to \dfrac{n}{m}$, and
\begin{equation}
  \label{eq:18}
  \eta^*_\worst\ge \min_\A\max_{W\in\W}\eta_W(\A)\to \dfrac{n}{m}.
\end{equation}
Comparing with (\ref{eq:14}) and (\ref{eq:18}), we get the following:
\begin{theorem}
When all $m$ basestations are identical for each of the $n$ users, the best worst-case competitive ratio scales as $n/m$ and the simple round-robin algorithm achieves it.  
\end{theorem}
\subsection{Case III: Arbitrary Weights}\label{subsec:GeneralWorstCase}
In this case, the weight matrix $W$ is assumed to contain arbitrary entries with no further assumptions. 
\subsubsection{Optimal offline algorithm} We first note that finding the optimal offline algorithm is challenging in this case. For deriving the competitive ratio, instead we upper bound the utility of the optimal offline algorithm.
Consider an algorithm $\A$ that makes an allocation $\M_\A(W)=\{M_j:\ 1\le j\le m\}$ on a weight matrix $W$. Clearly,
$TS(\M_\A(W),W)\le\sum_{j=1}^m\max_{i\in M_j}w_{ij}=\sum_{j=1}^mw_{i^*(j),j}$,
where $i^*(j)=\arg\max_{i\in M_j}w_{ij}$. The set of edges $\{(i^*(1),1),(i^*(2),2),\ldots,(i^*(m),m)\}$ is a matching in $G=(V_1\cup V_2,V_1\times V_2)$, the weighted complete bipartite graph of the allocation problem with $V_1=\{1,2,\ldots,n\}$ and $V_2=\{1,2,\ldots,m\}$ and weight of edge $(i,j)$ set to $w_{ij}$. So, if $\MWM(G)$ denotes maximum weight over all matchings in $G$, we have
\begin{equation}
  \label{eq:20}
  TS^*_\off\le TS(\M_\A(W),W)\le\sum_{j=1}^mw_{i^*(j),j}\le \MWM(G).
\end{equation}

\subsubsection{Bounds on the competitive ratio}
For an upper bound on the best worst-case competitive ratio, we use \eqref{eq:13} with a max-weight  online algorithm $\A_0$ that allocates each user to the basestation with maximum weight, i.e. user $i$ is alloted to basestation $j^*=\arg\max_jw_{ij}$. Let $\mw_i=w_{i,j^*}$ be the maximum weight of user $i$, and arrange the users such that $\mw_1\ge \mw_2\ge\cdots\ge \mw_n$. Now, 
$TS^*_\off\le \sum_{i=1}^m\mw_i$,
and 
$TS(\M_{\A_0}(W),W)\ge \dfrac{\sum_{i=1}^m\mw_i}{n}$,
since the maximum degree of any basestation in any allocation is $n$. Therefore,
\begin{equation}
  \label{eq:23}
  \eta^*_\worst\le\eta_\worst(\A_0)\le n.
\end{equation}
While the above upper bound looks rather pessimistic at first glance, it is, in fact, tight for the general case of arbitrary weights.  We have the following Theorem 
\begin{theorem}
\label{thm:worstcasegeneral}
For allocating $n$ users to $m$ basestations, the best worst-case competitive ratio scales as $n$ and an online algorithm that associates each user to the basestation with the largest weight achieves it.  
\end{theorem}
\vspace{-0.2in}
\begin{proof}
To prove the Theorem, we use \eqref{eq:13} with the set of weight matrices
$$\W=\{W_l=\begin{bmatrix}
\beta    &\e&\e&\cdots&\e\\
\beta^2&\e&\e&\cdots&\e\\
\vdots &\vdots&\vdots&\cdots&\vdots\\
\beta^l&\e&\e&\cdots&\e\\
\e        &\e&\e&\cdots&\e\\
\vdots &\vdots&\vdots&\cdots&\vdots\\
\e        &\e&\e&\cdots&\e
\end{bmatrix},\ 1\le l\le n\}$$
with $\beta>1$ and $\e<1$. For the optimal offline algorithm, by assigning the $l^{th}$ user to basestation $1$ and all other users to other basestations,
$TS(\M^*_\off(W_l),W_l)\ge \beta^l, 1\le l \le n$.
Now, consider an arbitrary online algorithm $\A$. We first show that if for any $W=W_j$, if a user is assigned by any online algorithm $\A$ to a basestation with weight $\e$, then $\eta_\worst(\A)\ge \infty$. Suppose user $l$ is assigned to a basestation with weight $\e$ for $W=W_j, l \le j$. Since $\A$ is online and $W_l$ matches with $W_j$ in the first $l$ rows, we have that $\A$ allocates the same basestation to user $l$ for $W=W_l$. So, we have
$$TS(\M_\A(W_l),W_l)\le \e+m\beta^{l-1}.$$
Therefore,
$$\eta_\worst(\A)\ge \dfrac{\beta^l}{\e+m\beta^{l-1}},$$
which can be made arbitrarily large by choosing a sufficiently large $\beta$ and  sufficiently small $\e$.

Therefore the only other choice for any online algorithm is to asign all users to basestation 1 (avoiding the allocation of any user to a basestation with weight $\e$), for which 
$$TS(\M_\A(W_n),W_n)\le \dfrac{\beta^n}{n}+\beta^{n-1}$$
resulting in
\begin{equation}
  \label{eq:24}
  \eta_\worst(\A)\ge \dfrac{\beta^n}{\beta^n/n+\beta^{n-1}}\to n,\text{ as }\beta\to\infty.
\end{equation}
Therefore we get
\begin{equation}
  \label{eq:25}
  \eta^*_\worst\ge \min_\A\max_{W\in\W}\eta_\worst(\A)\to n.
\end{equation}
\end{proof}

In summary, the worst case competitive ratio of online basestation allocation problems grow unbounded in the number of  users. Therefore, online algorithms for basestation allocation are not very promising in the worst case, and simple algorithms such as round-robin or max-weight allocation suffice. In practice, assuming  worst case for both the weights of each user and the order of user arrivals is somewhat pessimistic since there is inherent randomness in user arrivals. 
We study this 
more reasonable scenario in the next section,  where users arrive uniformly randomly, however, with arbitrary weights (including worst case).
We show that online algorithms have constant competitive ratio in the average case (uniformly random user arrivals), which has wide ranging implications for practical basestation allocation problems.

\section{Bounds on Average Case Competitive Ratio}
\label{sec:bounds-average-case}
In this section, we assume that weights of each user to any basestation are arbitrary (possibly worst case), while the order of user arrivals is uniformly random. Note that we are not making any assumptions on the statistics of any user's achievable rates.

\subsection{Case I: Identical basestations}
For average-case bounds on the competitive ratio, we first consider the
identical basestations case, where each user has equal weights to all
basestations. In this case, the weight matrix has a
constant value in each row, and the entire matrix is specified by the vector of weights in each row.
Let $\w=[w_1,w_2,\ldots,w_n]$ be the set of weights ordered such that $w_1\ge w_2\ge\cdots\ge w_n$. For considering average competitive ratio, we need to specify a distribution for the weight matrices. Assuming user arrivals to be uniformly random, the weight matrix $W$ is assumed to be
uniformly distributed over all permutations of
$\w=[w_1,w_2,\ldots,w_n]$, i.e. for any permutation $\pi$ of
$\{1,2,\ldots,n\}$, Pr$\{W=\pi(\w)\}=1/n!$, where $\pi(\w)=[w_{\pi(1)}\ w_{\pi(2)}\cdots
w_{\pi(n)}]$. Our bounds for the case of averaging over the
permutations with fixed weights will be
independent of $w_i$, and, therefore, the same bounds hold for any further averaging over the values of
the weights $w_i$.

\subsubsection{Two basestations}
When $m=2$ and $W=\pi(\w)$, there are two basestations and the optimal offline
algorithm allocates the maximum-weight user, user $\pi^{-1}(1)$, to basestation 2 (say),
while all the other users are allocated to basestation 1. To be
competitive, any online algorithm has the task of identifying the
maximum-weight user in an online fashion. Identifying the maximum
weight with high probability over random permutions of a sequence is
well-known as the secretary problem
\cite{Dynkin1963,Gilbert1966,Freeman1983}. Algorithms for the
secretary problem use the stopping rule method, which we adapt for
online allocation as follows.

Consider an online algorithm $\A(r)$ that works as follows:
\begin{enumerate}
\item Allocate the first $r$ users to basestation 1 and compute $T=\max\{w_{\pi(i)}:1\le i\le
r\}$. We will refer to $r$ as the number of test users.
\item The first subsequent user $i > r$ with weight $w_{\pi(i)}\ge T$
is allocated to basestation 2. All other subsequent users $j > i$ are allocated to basestation 1. 
\item If basestation 2 has no
user when $i=n$, the last user, user $n$, is allocated to
basestation 2. Otherwise, user $n$ is also allocated to basestation 1.
\end{enumerate}
Let $i^*$ denote the position of the user allocated to basestation
2. The online algorithm $\A(r)$ achieves a competitive ratio of $1$,
whenever the maximum weight user is allocated to basestation 2, or $\pi(i^*)=1$. Therefore,
\begin{equation}
  \label{eq:26}
  \eta_{\avg}(\A(r))\ge \pr(\pi(i^*)=1),
\end{equation}
where the probability is over a uniform choice of the permutation
$\pi$. Now, the maximum-weight user is allocated to basestation 2,
whenever that user arrives at position $i$ for $i>r$, and the maximum
weight among the previous $i-1$ users occurs within the first $r$ positions. So,
\begin{equation}
  \label{eq:27}
  \pr(\pi(i^*)=1)=\sum_{i=r+1}^n\pr(\pi(i)=1)\pr(\arg\max_{1\le j\le i}w_{\pi(j)}\le r\ |\ \pi(i)=1).
\end{equation}
Since $\pi$ is uniformly chosen, we have that $\pr(\pi(i)=1)=1/n$ for
all $i$, and 
$$\pr(\arg\max_{1\le j\le i}w_{\pi(j)}\le r\ |\ \pi(i)=1)=r/(i-1)$$ 
for a fixed $i$. So,
  \eqref{eq:27} simplifies as
  \begin{eqnarray}
    \label{eq:28}
    \pr(\pi(i^*)=1)&=&\sum_{i=r+1}^n\frac{1}{n}\frac{r}{i-1}\\
\label{eq:29}&=&\frac{r}{n}\sum_{i=r+1}^n\frac{1}{i-1}>\frac{r}{n}\log_e\frac{n}{r},
  \end{eqnarray}
where we have used the inequality
\begin{equation}
  \label{eq:34}
  \sum_{i=a}^b\frac{1}{i}>\int_a^{b+1}\frac{1}{x}dx=\log_e\frac{b+1}{a}
\end{equation}
for positive integers $a$ and $b$. Maximizing over the choice of $r$ in (\ref{eq:29}), the optimum $r$ is found to be
$n/e$, and using in \eqref{eq:26}, we get that $  \eta_\avg(\A(n/e))> \frac{1}{e}$.
So, we see that a constant average case competitive ratio is achievable by online algorithms as $n\to\infty$, while the worst case grows unbounded. The algorithm $\A(r)$ can be improved to achieve an even higher competitive ratio as described next. 
\subsubsection{Modified online algorithm for two basestations}
Consider an online algorithm $\tA(r)$ with $r$ test users that works as follows:
\begin{enumerate}
\item Allocate the first $r$ test users to basestation 1 and compute $T=\max\{w_{\pi(i)}:1\le i\le
r\}$. 
\item For $i=r+1,r+2,\ldots,n$, if $w_{\pi(i)}> T$,
  \begin{enumerate}
  \item Allocate user $i$ to basestation 2.
  \item Update $T=w_{\pi(i)}$. 
  \end{enumerate}
\item For $i=r+1,r+2,\ldots,n$, if $w_{\pi(i)}\le T$, allocate user $i$ to basestation 1. 
\end{enumerate}
Clearly, $\tA(r)$ will have an average competitive ratio strictly
greater than that of $\A(r)$. We now do an exact characterization of
this improvement.

Let $D$ denote the degree of basestation 2 after the running of algorithm $\tA(r)$. If the best user (user with maximum weight) arrives within the first $r$ positions, we will have $D=0$. For $D=1$, we require that the best user should arrive in position $i$ for some $r+1\le i\le n$ and, for this $i$, the best among the first $i-1$ users should arrive within the first $r$ positions. So, the event $D=1$ can be written as the disjoint union
$$\bigsqcup_{i=r+1}^n(\pi(i)=1)\cap(\arg\max_{1\le j\le i-1}w_{\pi(j)}\le r).$$
We now extend this argument to $D>1$. Let $D_i$ denote the degree of basestation 2 after user $i$ has been allocated. In particular, $D=D_n$, and the event $D_n=d$ satisfies the recursive relationship
\begin{equation}
  \label{eq:33}
  (D_n=d)\Leftrightarrow\bigsqcup_{i=r+d}^n(\pi(i)=1)\cap(D_{i-1}=d-1).
\end{equation}
\begin{prop}
  For the algorithm $\tA(r)$, the degree $D_n$ of basestation 2 after the arrival of $n$ users satisfies
  \begin{equation}
    \label{eq:35}
    \pr(D_n=d)\to\frac{1}{d!}\ \frac{r}{n}\left(\log_e\frac{n}{r}\right)^d,
  \end{equation}
for $r=\Theta(n)$, $d=o(n)$, as $n\to\infty$.
\label{prop:modif-online-algor}
\end{prop}
\begin{proof}
  The proof is by induction on $d$. The base case is $d=1$, which is easily shown using (\ref{eq:29}). Suppose, by the induction hypothesis, that (\ref{eq:35}) is true for $d-1$. Using (\ref{eq:33}), we get
  \begin{eqnarray}
\nonumber   \pr(D_n=d)&=&\sum_{i=r+d}^n\pr(\pi(i)=1)\pr(D_{i-1}=d-1),\\
\nonumber             &=&\sum_{i=r+d}^n\frac{1}{n}\ \frac{1}{(d-1)!}\ \frac{r}{i-1}\left(\log_e\frac{i-1}{r}\right)^{d-1}, \ \text{by induction,}\\
\label{eq:37}&=&\frac{r}{n}\ \frac{1}{(d-1)!}\sum_{i=r+d}^n\frac{1}{i-1}\left(\log_e\frac{i-1}{r}\right)^{d-1}.
  \end{eqnarray}
For $n\to\infty$, $r=\Theta(n)$ and $d=o(n)$, the summation above simplifies as
\begin{equation}
  \label{eq:36}
  \sum_{i=r+d}^n\frac{1}{i-1}\left(\log_e\frac{i-1}{r}\right)^{d-1} \to \frac{1}{d}\left(\log_e\frac{n}{r}\right)^d.
\end{equation}
 Using \eqref{eq:36} in \eqref{eq:37}, the proof is complete.
 \end{proof}
 Using Proposition \ref{prop:modif-online-algor}, the average
 competitive ratio of $\tA(r)$ can be lower-bounded, and this is captured in the following
 theorem.
 \begin{theorem}
   Let $\alpha$ be a constant fraction, and let $r=\alpha n$ be the number of test users in the online algorithm $\tA(r)$ for allocating $n$ users to two basestations. As the number of  users $n\to\infty$, the average competitive ratio
   of $\tA(r)$ satisfies
   \begin{equation}
     \label{eq:38}
     \eta_\avg(\tA(\alpha n)) \ge \sum_{d=1}^{n(1-\alpha)}\frac{1}{d}\    \frac{1}{d!}\ \alpha\left(\log_e\frac{1}{\alpha}\right)^d, 
  \end{equation}
where $d$ in the summation represents the degree of Basestation 2 after $\tA(r)$ completes. 
Summing this expression upto $d\le d_{\max}=10$, we get that the highest lower bound is obtained at $\alpha=0.22$, and $\eta_\avg(\tA(0.22 n)) \ge 0.517$
by numerical computations.
 \label{thm:modif-online-algor}
 \end{theorem}
 \begin{proof}
   Since
 $$\eta_\avg(\tA(r)) \ge \sum_{d}\frac{1}{d}\pr(D_n=d),$$
 the proof is complete by Proposition \ref{prop:modif-online-algor}.
 \end{proof}
So, we see that online algorithms can be more than 50\% competitive
for allocation to two basestations in the average case. We now extend
the modified algorithm to an arbitrary number of basestations.
\subsubsection{Arbitrary number of basestations}\label{sec:arbbase}
When the number of basestations $m$ is arbitrary and the weight matrix
is $W=\pi(\w)$, the optimal offline algorithm allocates each of the top $m-1$
users arriving at positions $\pi^{-1}(i)$, $1\le i\le m-1$, respectively, to
basestation $i$, while all other users are allocated to basestation
$m$. The optimal offline utility is given by 
$$w_1+w_2+\cdots+w_{m-1}+\frac{1}{n-(m-1)}\sum_{i=m}^nw_i.$$
To be competitive, an online algorithm has the task of identifying
$m-1$ weights with maximal sum in an online fashion. This is a generalization of
the secretary problem (called the $k$-secretary problem) and has been considered in
\cite{LleinbergK-Secretary2005}\cite{babaioff2007k-Secretary}. In \cite{LleinbergK-Secretary2005}, a randomized algorithm achieving an
expected competitive ratio (on the sum of $m-1$ weights) of $O(1-1/\sqrt{m-1})$ 
has been provided. In \cite{babaioff2007k-Secretary}, simpler algorithms achieving better competitive ratios for small $m$ have been provided. 



We present an algorithm that has some elements of the virtual algorithm in \cite{babaioff2007k-Secretary} combined with the allocation algorithm for two basestations considered earlier in this section. The proposed algorithm, denoted $\A_m(r)$, with $n$ users to be allocated among $m$ basestations, proceeds as follows:
\begin{enumerate}
\item Allocate the first $r$ test users to basestation 1, and compute the $(m-1)$-th best weight, denoted $T$, among the first $r$ users. 
\item Set $j=2$.
\item For $i=r+1,r+2,\ldots,n$, if $w_{\pi(i)}> T$,
  \begin{enumerate}
  \item Allocate user $i$ to basestation $j$.
  \item Update $T$ as the $(m-1)$-th best user seen so far.
  \item Set $j=j+1$. If $j>m$, set $j=2$. 
  \end{enumerate}
\item For $i=r+1,r+2,\ldots,n$, if $w_{\pi(i)}\le T$, allocate user $i$ to basestation 1. 
\end{enumerate}
The algorithm $\A_m(r)$ allots users, arriving after position $r$, in a round-robin fashion to basestations 2 through $m$, if their weight is within the top $(m-1)$ weights seen so far. Otherwise, the user is allotted to basestation 1. We say that a user is \emph{selected}, if the user is allotted to basestation $j$ for $2\le j\le m$. Let $S_n$ denote the number of selected users. In the ensuing analysis, we determine the probability, over all uniformly random permutations $\pi$, that $S_n=d$. We will assume that $m=o(n)$ and $d=o(n)$ throughout.

For a permutation $\pi$ on the $n$ users, let 
$ \bp(i)=|\{j:w_j\ge w_i,\ 1\le j\le i\}|$
denote the number of users in the first $i$ positions with weight greater than or equal to $w_i$. Note that the weight of user $i$ is within the top $\bp(i)$ weights seen so far, and $\bp(i)\in\{1,2,\ldots,i\}$. So, for $i>r$, user $i$ is selected if $\bp(i)\le m-1$, and not selected otherwise. 

As shown in \cite{Renyi1970} (also see \cite[Chapter 3, Problem 16]{Lovasz2007}), there is a bijection between the set of all permutations $\pi$ on $n$ users and the vectors $[\bp(1)\ \bp(2)\cdots\bp(n)]$, $1\le\bp(i)\le i$, for each $i$. So, selecting a permutation $\pi$ uniformly at random is equivalent to selecting $\bp(i)$ independently for $i=1,2,\ldots,n$ with $\bp(i)$ being uniform in $\{1,2,\ldots,i\}$. Using this bijection, it readily follows that
\begin{equation}
  \label{eq:40}
  \pr(S_n=d)=\sum_{r+1\le i_1<i_2<\atop\cdots<i_d\le n}\left(\prod_{i\in\{i_1,i_2,\ldots,i_d\}}\frac{m-1}{i}\right)\left(\prod_{r+1\le i\le n \atop \ i\notin\{i_1,i_2,\ldots,i_d\}}1-\frac{m-1}{i}\right),
\end{equation}
where the first product term represents the $d$ selected users for which $\bp(i)\le m-1$, and the second product term counts all the other $n-d-r-1$ non-selected users for which $\bp(i) > m-1$. 
  
By careful manipulation of the product term in \eqref{eq:40} (shown in Appendix \ref{sec:asympt-estim-prs_n=d}), we obtain
\begin{equation}
  \label{eq:41}
  \pr(S_n=d)=\frac{r(r-1)\cdots(r-m+2)}{n(n-1)\cdots(n-m+2)}(m-1)^d\sum_{r+1\le
    i_1<i_2<\atop\cdots<i_d\le n}\left(\prod_{i\in\{i_1,i_2,\ldots,i_d\}}\frac{1}{i-(m-1)}\right).
\end{equation}
Using techniques similar to those in the proof of Proposition
\ref{prop:modif-online-algor}, the summation in \eqref{eq:41} can be
simplified as shown in Appendix \ref{sec:asympt-estim-prs_n=d}, and we get
\begin{equation}
  \label{eq:42}
  \pr(S_n=d)\to \left(\frac{r}{n}\right)^{m-1}\frac{1}{d!}\left((m-1)\log_e\frac{n}{r}\right)^d
\end{equation}
for $m=o(n)$ and $d=o(n)$. Since the selected users are allocated in a round-robin fashion, we get the following:
\begin{theorem}
\label{thm:arbweights}
   Let $\alpha$ be a constant fraction, and let $r=\alpha n$ be the number of test users in the online algorithm $\A_m(r)$ for allocating $n$ users to $m$ basestations. As the number of  users $n\to\infty$, the average competitive ratio of $\A_m(r)$ satisfies
   \begin{equation}
     \label{eq:43}
     \eta_\avg(\A_m(\alpha n)) \ge \sum_{d=m-1}^{n(1-\alpha)}\frac{1}{\ceil*{d/(m-1)}}\ \alpha^{m-1}\frac{1}{d!}\left((m-1)\log_e\frac{1}{\alpha}\right)^d, 
   \end{equation}
where $d$ in the summation represents the number of selected users after $\A_m(r)$ completes. Summing for $d\le d_{\max}=10$, we get that the highest lower bound is obtained at $\alpha=0.22$, 
$$\eta_\avg(\A_m(0.22 n)) \ge 0.46$$
for $2\le m\le 20$ by numerical computations.
\label{thm:arbitr-numb-basest}
\end{theorem}
\begin{proof}
  Since the allocation to basestations 2 through $m$ is done in a round-robin fashion, the degree of each basestation $j$ is at most $\ceil*{d/(m-1)}$ for $2\le j\le m$. After that, the proof proceeds as in the case of Theorem \ref{thm:modif-online-algor} using (\ref{eq:42}). 
\end{proof}
\subsection{Case II: Arbitrary Weights} \label{subsec:GeneralRandom}
We move on to the case where the weights of a user to each basestation
are arbitrary. In this case, we fix a weight matrix $W$ and consider
averaging over all permutations $\pi$ of the rows of $W$, i.e. the
users arrive in an arbitrary order specified by $\pi$. The permuted
version of $W$ is denoted $\pi(W)$. Since our
bound will be independent of the weights, the same bound holds upon
further averaging over an arbitrary distribution on the weights themselves.

When the weights are arbitrary, the optimal offline timesharing
utility can be upper-bounded by the max-weight matching in the complete
weighted bipartite graph $G=(V_1\cup V_2,V_1\times V_2)$ as shown in
\eqref{eq:20}. In $G$, the edge $(i,j)$ connects user $i$ to
basestation $j$ and has weight $w_{ij}$. 

Online algorithms to find max-weight matchings in bipartite graphs have
been studied by several authors. See \cite{Korula2009} and references
therein for more details. The averaging in these online matching
algorithms is over all orderings of vertices in one partition, and
this is the same as the distribution $\pi(W)$ considered here in the
basestation allocation problem. 

We now make use of online max-weight matching algorithms
to develop a randomized online allocation algorithm and compute its
expected competitive ratio under the time-sharing utility. 
The SAMPLEANDPRICE algorithm proposed in \cite{Korula2009} for online max-weight matching with input bipartite graph $(G(L \cup R), E)$,  $L$ and $R$ are left and right vertices, and $E$ is the weighted edge set,  where when a vertex $\ell \in L$ is seen, all edges incident to $\ell$ are revealed, together with their weights. The
algorithm immediately decides to either match $\ell$ to an available vertex of $R$, or never matches $\ell$.


\begin{table}
\begin{tabular}{r l}
\hline
& \textbf{SAMPLEANDPRICE($G(|L|, R)$}\\
\hline
1	&	\text{$k \leftarrow Binomial(|L|, p)$} \\
2	&	\text{Let $L'$ be the first $k$ vertices of $L$}\\
3	&	\text{$M_1 \leftarrow$ GREEDY($G(L' \cup R]$)}\\
4	&	\text{For each $r \in R$:}\\
5	&	\ \ \text{Set $price(r)$ to be the weight of the edge incident to $r$ i $M_1$ }\\
6	&	$M \leftarrow \Phi $\\
7	&	\text{For each subsequent $\ell \in L \backslash L'$, :}\\
8	&	\ \ \text{Let $e = (\ell, r)$ be the highest-weight edge such that $w(e) \ge price (r)$ }\\
9	&	\ \ \text{If $M\cup e$ is a matching, accept $e$ for $M$}\\
\hline
\end{tabular} 
\end{table}

\begin{table}
\begin{tabular}{r l}
\hline
& \textbf{GREEDY ($G(L\cup R, E)$}\\
\hline
1	&	\text{Sort edges of $E$ in decreasing order of weight.} \\
2	&	\text{Matching $M \leftarrow \Phi$}\\
3	&	\text{For each edge $e \in E$, in sorted order}\\
4	&	\text{If $M\cup e$ is a matching }\\
5	&	\ \ $M \leftarrow M \cup e $\\
6	&	\text{Return $M$}\\
\hline
\end{tabular} 
\end{table}

\begin{lemma}\label{lem:Korula} \cite{Korula2009} For $p=\frac{1}{2}$, the competitive ratio of SAMPLEANDPRICE algorithm is $8$.
\end{lemma}

Now we propose an online algorithm called HideAndSeek using the SAMPLEANDPRICE algorithm to maximize the timesharing utility.

${\bf HideAndSeek \ Algorithm}$
\begin{enumerate}
\item Let $j_0\in\{1,2,\ldots,m\}$ be chosen uniformly at random.
\item Use SAMPLEANDPRICE algorithm $\O$  to find the max-weight matching in the graph 
  \begin{equation}
    \label{eq:22}
G_{-j_0}=(V_1\cup (V_2\setminus\{j_0\}) ,V_1\times (V_2\setminus\{j_0\}))    
  \end{equation}
with basestation $j_0$ deleted from the original graph $G$. Denote
this matching as $\O(G_{-j_0})$.
\item Allocate the users in the matching $\O(G_{-j_0})$ to the
  $m-1$ basestations other than $j_0$, while all other users are
  allocated to basestation $j_0$ itself.
\end{enumerate}
We call this algorithm  HideAndSeek since it first randomly hides one basestation and then seeks an online max-weight matching for the rest of the basestations.

\begin{theorem} With randomized input the HideAndSeek  algorithm has competitive ratio $< \frac{m-1}{8m}$.
\end{theorem}
\begin{proof}
Denote the matching that is output by the online max-weight algorithm $\O$ run on a bipartite
graph $G$ as $\O(G)$ and denote its weight by $\MWM_{\O}(G)$. Now, 
\begin{equation}
  \label{eq:21}
  E[\MWM_{\O}(G_{-j_0})]\ge \frac{m-1}{m}\MWM_{\O}(G),
\end{equation}
since $\O(G)$ without the edge to basestation $j_0$ is a possible
matching in $G_{-j_0}$, and $j_0$ is chosen uniformly at random. 
Thus for the HideAndSeek algorithm, the expected utility $E[TS_{HS}(\pi(W))]$ is lower bounded by  $E[\MWM_{\O}(G_{-j_0})]$, by not  considering the contribution of the users associated with the randomly chosen basestation $j_0$. Hence, we have
\begin{eqnarray}\nn
  E[TS_{HS}(\pi(W))]& \ge &E[\MWM_{\O}(G_{-j_0})],\\ \nn
  &\stackrel{(a)}\ge& \frac{1}{8}E[\MWM_{\O}(G_{-j_0})], \\    \label{eq:31}
  &\stackrel{(b)}\ge& \frac{m-1}{8m}\MWM(G),
  \end{eqnarray}
where $(a)$ follows since $\O$ is the SAMPLEANDPRICE algorithm that is $8$-competitive (Lemma \ref{lem:Korula}), while $(b)$ follows from (\ref{eq:21}). Since the optimal offline utility is upper-bounded by $\MWM(G)$, the expected average competitive ratio
of the algorithm $\A$ is given as
\begin{equation*}
  \label{eq:32}
  E[\eta_{\avg}(\A)]\ge \frac{m-1}{8m}.
\end{equation*}
\end{proof}
\vspace{-0.2in}
{\it Discussion:} To summarize, in this section we have shown that with the uniformly random user arrival order, online algorithms are a good candidate for practical use since they have constant competitive ratio with the offline algorithm, that provides the {\it absolute} benchmark to the overall sum-rate performance. The online algorithms presented in this section tend to allocate only few {\it strong} users to most of the basestations and load one basestation with most of the {\it weak} users. Even though this is not desirable from a selfish user's perspective, such an allocation maximizes the sum-rate close to what an optimal offline algorithm can, knowing everything in future.

\subsection{Practical Implications}\label{sec:implication} Most often in practice, any incoming user is assigned to the basestation for which it has the maximum weight (achievable rate). As we saw in Section \ref{subsec:GeneralWorstCase}, this max-weight association was optimal in the worst case input case, where the competitive ratio grows as the number of users
$n$. In the more reasonable scenario of randomized user arrival order, even for just two basestations, consider a user weight matrix 
$$W=\begin{bmatrix}
\beta    &\beta^2&\cdots&\beta^n\\
\sqrt{\beta}&\beta &\cdots&\beta^{n-1}\\
\end{bmatrix},$$
where each column represents the weight of any user to the two basestations. Now if users arrive in any order (permutations of columns), the max-weight association algorithm's utility will be at the maximum $\sum_{i}^n\beta^i/n$, while the optimal offline algorithm's utility is at least $\beta^n$. Thus, for large $\beta$ the competitive ratio of the max-weight association is still $n$, if only the user arrival order is randomized, while the weights of each user are still worst case. 
As we proved in Section \ref{subsec:GeneralRandom}, however,
the competitive ratio of HideAndSeek algorithm is a constant. Thus, the algorithms presented in Section \ref{subsec:GeneralRandom} are far better alternatives to the popular max-weight basestation association.

\section{Reassignments}
In previous sections, we have assumed that once a user is allocated to a basestation that allocation is irrevocable. In practice, this is typically the case, since reassignments are costly. However, from the utility point of view there could be a big advantage even if a small number of users are allowed to be reassigned. In
 this section, we want to identify how much improvement can be obtained in terms of competitive ratio if some users are allowed to be reassigned.

\subsection{Identical basestations}

\begin{theorem} If on user $i$'s arrival any one of the previously assigned $1,\dots, i-1$ users  is allowed to be reassigned, then the competitive ratio is $1$, i.e. online algorithm performs as good as an optimal offline algorithm, even with adversarial inputs. 
\end{theorem}
\begin{proof} Recall from Section \ref{subsec:optoffline} that the optimal offline algorithm is to allocate top $m-1$ users to $m-1$ basestations with no two sharing a basestation, and the $n-(m-1)$ weakest users are allocated to a single basestation. Thus, a simple online algorithm that, on each user arrival, assigns or reassigns the current weakest user  (depending on when the current weakest user arrived) to the one basestation designated for the $n-(m-1)$ weakest users, achieves the optimal offline utility.
\end{proof}

If suppose rather than allowing any one of the previous users to be reassigned, if only the last arrived user  is allowed to be reassigned then we show in the next Theorem that it is equivalent to not having allowed any reassignment and the competitive ratio is $\approx \frac{n}{m}$. 

\begin{theorem} If on user $i$'s arrival only the previously assigned user $i-1$ is allowed to be reassigned then the competitive ratio is $\frac{n}{m}$. 
\end{theorem}
\begin{proof} Recall the case of no reassignments being allowed in Section \ref{subsec:LBId}. 
The set of bad weight matrices $ \W$ defined in (\ref{eq:15}) continue to be bad, even when only the last arrived user  is allowed to be reassigned, if a $0$ is inserted between any two non-zero entries ($\beta^{i}$ and $\beta^{i+1}$ for $i=1,\ldots,l$) of the weight matrices. With these modified weight matrices, the result follows immediately. \end{proof}

\subsection{Unequal Weights}
\begin{theorem} With adversarial inputs,  if on user $i$'s arrival any one of the previously assigned  $1,\dots, i-1$ users is allowed to be reassigned then the competitive ratio $>n/2$. 
\end{theorem}
\begin{proof}
For $m=2$, consider two bad inputs: 
\[W_1 = 
\left[\begin{array}{ccccccc} b & \dots &b   &  a  & \dots & a & x 
\\
\undermat{n/2}{a & \dots & a}  &  \undermat{n/2-1}{c & \dots & c} & a 
\end{array}\right], W_2 = 
\left[\begin{array}{ccccccc} b & \dots &b   &  a  & \dots & a & a 
\\
\undermat{n/2}{a & \dots & a}  &  \undermat{n/2-1}{c & \dots & c} & x 
\end{array}\right],\]
\vspace{0.1in}

where $a << b, a<< c, b \approx c$, and $x >> b, x>> c$. 
Since, with $W_1$, weight of user $n$ to basestation $1$, $x$, is extremely large compared to all other weights, optimal offline algorithm will assign the first $n-1$ users to basestation $2$ and allocate only user $n$ to basestation $1$, while exactly the opposite happens with $W_2$.

An online algorithm, at slot $n-1$ not knowing whether in the $n$-th slot the weights are $(x, a)$ or $(a,x)$, would have allocated $n/2-1$ users to at least one of the two basestations. Without loss of generality, let basestation $1$ have  $n/2-1$ users at the end of slot $n-1$. Then if $W_1$ is the actual sequence, i.e. $(x, a)$ is the weight of the $n^{th}$ user, the utility of the online algorithm is at most $x/(n/2-1)$.

Thus, $\eta > \max \left\{ 
\frac{
TS({\cal M}^*_{off}(W_1)}{TS(A(W_1)
}, 
\frac{
TS({\cal M}^*_{off}(W_2)}{TS(A(W_2)
}
\right\}$. Since $TS({\cal M}^*_{off}(W_1) = TS({\cal M}^*_{off}(W_1) \approx x$, while at least one of  $TS(A(W_1)$ or $TS(A(W_1)$ is less than $\frac{x}{n/2-1}$, hence $\eta > n/2-1$.
\end{proof}

\begin{theorem}\label{thm:ReRa} With randomized input, if on user $i$'s arrival any one of the previously assigned  $1,\dots, i-1$ users is allowed to be reassigned then the competitive ratio $< \frac{m-1}{2m}$. 
\end{theorem}
To prove this Theorem, consider a  greedy offline algorithm (called GOA) for solving the offline max-weight matching problem, that greedily adds the heaviest edge possible to the current matching and stops when no further edge can be added. The performance of GOA is easy to quantify as described in the next Lemma.

\begin{lemma}\label{lem:greedyoffline} GOA is $2$-competitive.
\end{lemma}
\begin{proof}
Let $M^*$ be the optimal matching of graph $G$ and $M$ be the matching obtained using
the greedy algorithm. Let $(u, v)$ be an edge chosen by the greedy algorithm that is not
present in $M^*$. If $M^*$ does not contain this edge, then there should have been two other
edges between $u$ and some vertex $v'$, and $v$ and some vertex $u'$ in $M^*$. 
Since the greedy algorithm chose $(u, v)$,
$w_{uv} \ge w_{uv'}$  and  $w_{uv} \ge w_{u'v}$. Thus, $2 w_{uv} \ge w_{uv'} + w_{u'v}$, and hence 
$w_{uv} \ge \frac{w_{uv'} + w_{u'v}}{2}$. Thus, accounting for all egdes of $M$ the sum of weights of all edges in $M$ is at least half the sum of the weights of edges in $M^*$.
\end{proof}

Now we define an online greedy algorithm $O_r$ for the max-weight matching problem that is allowed to reassign or delete one of the previously assigned edges at each time $j, j=1,\dots,n,$ if required.  At time $j-1$, let the weights of currently assigned edges to the $m$ basestations  of set $V_2$ be $O_{k}^{j-1}, k=1,\dots, m$. 
Then at time $j$, from the set of $m$ newly arrived edges $w_{j,k}, k=1,\dots,m$, $O_r$ finds the index $k$ 
of the heaviest weight $w_{j,k} > O_{k}^{j-1}$ and assigns user $j-1$ to basestation $k$ by removing the previously 
assigned user to basestation $k$ if any.

\begin{lemma}\label{lem:reassign} If at time $j$, any one of the previously assigned edges at time $1,\dots, j-1$ is allowed to be reassigned or deleted, then the output of $O_r$ is identical to GOA at time $j$, even though $O_r$ is online. 
\end{lemma}

\begin{proof} The proof proceeds by induction. It is easy to verify the claim for the base case $j=2$. Assume the hypothesis is true for time $j$.
Let the output of GOA on an offline input from time $1, \dots, j$, and time $1, \dots, j+1,$ be 
$G_{k}^{j}, k=1,\dots, m$ and $G_{k}^{j+1}, k=1,\dots, m$, respectively. Note that $G_{k}^{j+1} \ne G_{k}^{j}$ if and
only if the weight of at least one edge among the $m$ newly arrived edges $w_{j+1,k}, k=1,\dots,m,$ is more than the corresponding element in $G_{k}^{j}, k=1,\dots, m$. Let the index of the heaviest edge for which $w_{j+1,k} > G_{k}^{j}, k=1,\dots, m$ be $k'$. Then $G_{k}^{j+1} =  \{G_{k}^{j} \cup w_{j+1,k'}\}, k\ne k', k=1,\dots, m$. More importantly note that the same is true for the online algorithm $O_r$ that deletes the $k'$ edge at time 
$j$ of $O_{k}^{j}$ and adds the heaviest edge from the newly arrived edges at time $j+1$ to $O_{k}^{j}$ if $w_{j+1,k} > O_{k}^{j}$, and $O_{k}^{j+1} =  \{O_{k}^{j} \cup w_{j+1,k'}\}, k\ne k', k=1,\dots, m$. 
This proves the claim, since we assumed that $G_{k}^{j} = O_{k}^{j}$.
\end{proof}

To prove Theorem \ref{thm:ReRa} we propose the following algorithm. 
Consider an algorithm ${\cal R}$ that does the following:
\begin{enumerate}
\item Let $j_0\in\{1,2,\ldots,m\}$ be chosen uniformly at random.
\item Use online algorithm $O_r$ to find the max-weight matching in the graph 
  \begin{equation}
    \label{eq:22}
G_{-j_0}=(V_1\cup (V_2\setminus\{j_0\}) ,V_1\times (V_2\setminus\{j_0\}))    
  \end{equation}
with basestation $j_0$ deleted from the original graph $G$. Denote
this matching as $\O(G_{-j_0})$.
\item Allocate the users in the matching $\O(G_{-j_0})$ to the
  $m-1$ basestations other than $j_0$, while all other users are
  allocated to basestation $j_0$ itself.
\end{enumerate}

\begin{proof}[Proof of Theorem \ref{thm:ReRa}] 
Denote the matching that is output by the online algorithm $O_r$ run on a bipartite
graph $G$ as $O_r(G)$ and denote its weight by $\MWM_{O_r}(G)$. Now, 
\begin{equation}
  \label{eq:221}
  E[\MWM_{O_r}(G_{-j_0})]\ge \frac{m-1}{m}\MWM_{O_r}(G),
\end{equation}
since $\O(G)$ without the edge to basestation $j_0$ is a possible
matching in $G_{-j_0}$, and $j_0$ is chosen uniformly at random. Using
\eqref{eq:221}, similar to \eqref{eq:31}, we see that
\begin{equation}
  \label{eq:331}
  E[TS_{{\cal R}}(\pi(W))] \ge E[\MWM_{O_r}(G_{-j_0})]\ge \frac{1}{2}E[\MWM_{O_r}(G_{-j_0})]\ge \frac{m-1}{2m}\MWM(G),
\end{equation}
since $O_r$ is a $2$-competitive algorithm as shown in Lemma \ref{lem:greedyoffline}, \ref{lem:reassign}.
 Since the optimal offline utility is upper-bounded by $\MWM(G)$, the expected average competitive ratio
of the algorithm ${\cal R}$ is given as
$E[\eta_{\avg}(\A)]\ge \frac{m-1}{2m}$.

\end{proof}
{\it Discussion:} In this section, we quantified the gain in terms of competitive ratio when some users can be reassigned. We showed that competitive ratio can be improved $4$ times even if only one user is allowed to be reassigned under the randomized user arrival order scenario, which is significant. Even though reassignments are costly, we showed that there is significant value in doing so.
\vspace{-.2in}
\section{Simulation Results}
We first consider the case of identical basestations and plot the competitive ratio obtained by our 
$k$-secretary based algorithm ${\cal A}_m(r)$ descibed in Section \ref{sec:arbbase}, and compare it with the max-weight algorithm. For the case of identical base stations, max-weight algorithm  associates each incoming user to the basestation with the minimum number of associated users. 
We plot the competitive ratio for different values of $r$, the number of test users that are used to determine the threshold for allocation for $m=10$ basestations by varying the number of users $n$ in Fig. \ref{fig:ksec}. We plot for both the idealistic scenario of i.i.d. uniform rates between $[0,10]$, and a correlation model, 
where mean rate for each user from basestation $1$ is $10$ and mean rate for each user from any other basestation is $5$. We see that the competitive ratio is close to $1$ as $n$ grows large for our algorithm, which is better than the worst case bound of  Theorem \ref{thm:arbweights} for both the cases. The max-weight algorithm has a competitive ratio of around $2$ for the i.i.d. uniform rates and competitive ratio of $2.5$ for the correlated model.  
\begin{figure}
\centering
\includegraphics[width=5in]{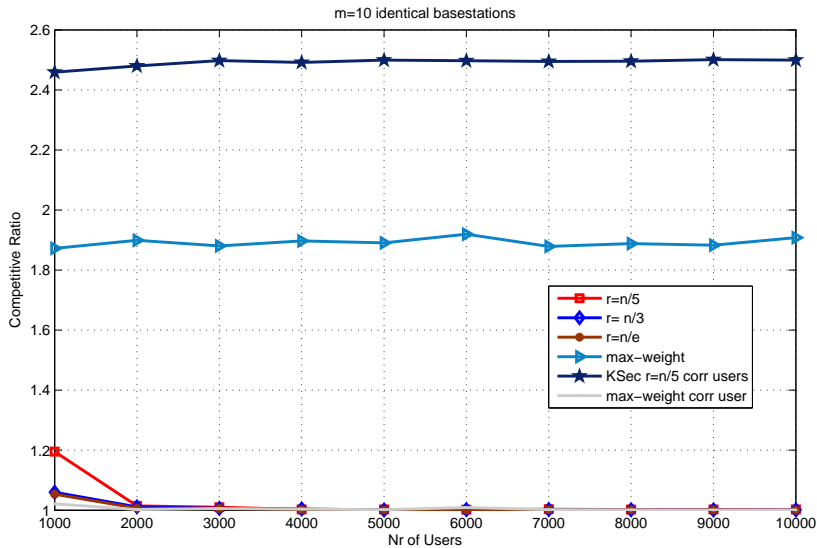}
\caption{Competitive ratio performance of $k$-Secretary problem based algorithm.}
\label{fig:ksec}
\end{figure} 
\begin{figure}
\centering
\includegraphics[width=4in]{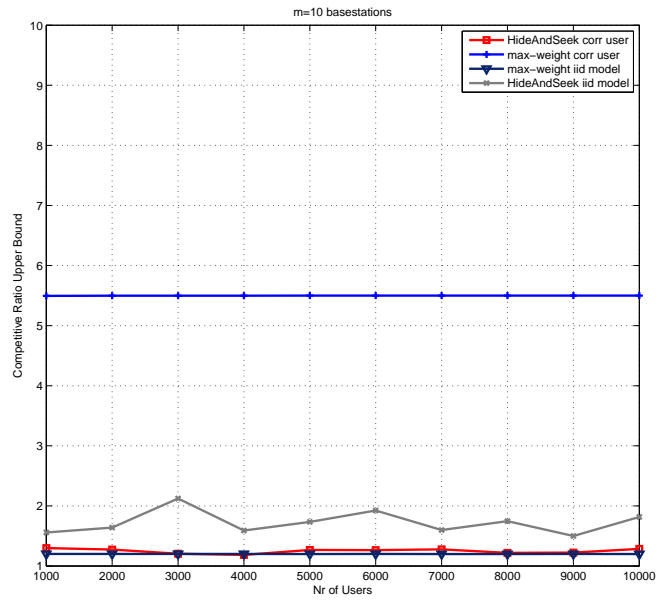}
\caption{Competitive ratio comparison of HideAndSeek and max-weight algorithm with correlated input.}
\label{fig:arbweightsworst}
\end{figure}
\begin{figure}
\centering
\includegraphics[width=4in]{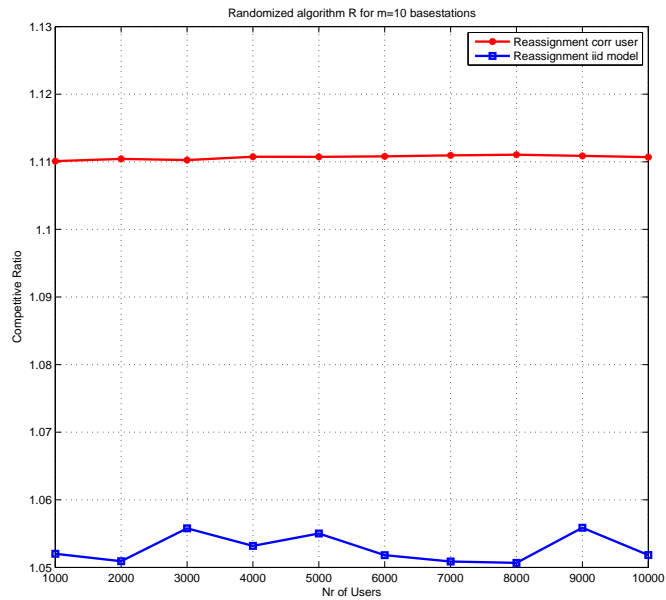}
\caption{Competitive ratio performance of Reassignment Algorithm.}
\label{fig:reassign}
\end{figure}

Next, we consider the arbitrary weights case, and plot the competitive ratio for the HideAndSeek algorithm and compare it with the max-weight algorithm in Fig. \ref{fig:arbweightsworst} for $m=10$. 
We consider the correlation model defined above, and the i.i.d. uniform rate model between $[0,10]$. From Fig. \ref{fig:arbweightsworst}, we see that even for this simple correlation model,  the competitive ratio of HideAndSeek algorithm is five times better than the max-weight algorithm, while in the i.i.d. model we see that competitive ratio of the HideAndSeek algorithm and the max-weight algorithm are similar, and actually max-weight algorithm outperforms the HideAndSeek algorithm by a little bit.
As discussed before, the competitive ratio of the max-weight algorithm depends on the number of users $n$ and base stations $m$. For instance, using the example given in Section \ref{sec:implication},  one can actually show that the competitive ratio of max-weight algorithm worsens as $n$ grows large, while the competitive ratio of the HideAndSeek algorithm is actually better than the worst case bound of $8$ we derived. Finally, in Fig. \ref{fig:reassign}, we plot the competitive ratio of the reassignment algorithm ${\cal R}$, where at any time one of the previously assigned user can be reassigned, for both the correlation model (described above) and the i.i.d. model for $m=10$. 
Comparing with Fig. \ref{fig:arbweightsworst}, we see that the performance of 
the reassignment algorithm is better than HideAndSeek algorithm for both the cases as expected, but more importantly it is better than the max-weight algorithm for the i.i.d. model as shown in 
Fig. \ref{fig:arbweightsworst}.

\section{Conclusions}
In this paper, we took first steps in understanding the fundamental sum-rate performance of basestation association problem when no assumptions are made about user statistics. 
We first showed that with the worst case input, the competitive ratio of any online algorithm grows linearly with the number of users and hence is too pessimistic. Then, we restricted ourselves to the more realistic case of random user arrival order and then derived online algorithms with constant competitive ratio. This is a significant result since the online algorithm only pays a constant penalty with respect to the offline algorithm that are provided with all the user information in advance. More importantly, we showed that the often used max-weight basestation allocation is very bad in terms of competitive ratio and the algorithms proposed in this paper provide a better alternative. 

The next logical step of this work is to include the functionality of allowing users to exit and not confining the system to a fixed time duration. The queuing theoretic model is the most natural for this purpose, where at each time some users arrive and their achievable rates are revealed upon arrival. Each user exits the system once its service is done, and the problem is to assign each user to one basestation to maximize the long term sum-rate of the system. This is part of on-going work.

\appendices
\section{Asymptotic Estimation for $\pr(S_n=d)$ in (\ref{eq:40})}
\label{sec:asympt-estim-prs_n=d}
We begin with the equation
\begin{equation}
  \label{eq:55}
  \pr(S_n=d)=\sum_{r+1\le i_1<i_2<\atop\cdots<i_d\le n}\left(\prod_{i\in\{i_1,i_2,\ldots,i_d\}}\frac{m-1}{i}\right)\left(\prod_{r+1\le i\le n \atop \ i\notin\{i_1,i_2,\ldots,i_d\}}1-\frac{m-1}{i}\right).  
\end{equation}
Each term in the summation above can be simplified as
\begin{equation}
  \label{eq:56}
  \prod_{i\in\{i_1,i_2,\ldots,i_d\}}\frac{m-1}{i}\prod_{r+1\le i\le n \atop \ i\notin\{i_1,i_2,\ldots,i_d\}}1-\frac{m-1}{i}
= (m-1)^d\prod_{r+1\le i\le n}\frac{1}{i}\prod_{r+1\le i\le n \atop \ i\notin\{i_1,i_2,\ldots,i_d\}}(i-m+1).
\end{equation}
Multiplying and dividing the RHS of (\ref{eq:56}) by $\prod_{r+1\le i\le n}(i-m+1)$ and canceling common terms, we further get
\begin{equation}
  \label{eq:57}
  \prod_{i\in\{i_1,i_2,\ldots,i_d\}}\frac{m-1}{i}\prod_{r+1\le i\le n \atop \ i\notin\{i_1,i_2,\ldots,i_d\}}1-\frac{m-1}{i}
= (m-1)^d\frac{r(r-1)\cdots(r-m+2)}{n(n-1)\cdots(n-m+2)}\prod_{i\in\{i_1,i_2,\ldots,i_d\}}\frac{1}{i-(m-1)}.
\end{equation}
Using (\ref{eq:57}) in (\ref{eq:55}), we get
\begin{equation}
  \label{eq:44}
\pr(S_n=d)=\frac{r(r-1)\cdots(r-m+2)}{n(n-1)\cdots(n-m+2)}(m-1)^d\sum_{r+1\le
    i_1<i_2<\atop\cdots<i_d\le n}\left(\prod_{i\in\{i_1,i_2,\ldots,i_d\}}\frac{1}{i-(m-1)}\right).  
\end{equation}
Letting
$A_d(t,n)=\sum_{t+1\le
    i_1<i_2<\atop\cdots<i_d\le n}\left(\frac{1}{i_1i_2\cdots i_d}\right)$,  
we see that the summation in (\ref{eq:44}) can be written as $A_d(r-(m-1),n-(m-1))$.
\begin{prop}
  \begin{equation}
    \label{eq:46}
A_d(t,n)=\frac{1}{d!}\left(\log\frac{n}{t}\right)^d\pm O\left(\frac{\left(\log\frac{n}{t}\right)^{d-1}}{t}\right).    
  \end{equation}
\label{prop:app}
\end{prop}
\vspace{-0.5in}
\begin{proof}
  The proof is by induction on $d$. The base case is $d=1$, which was shown in the proof of Proposition \ref{prop:modif-online-algor}. By induction hyptothesis, suppose that (\ref{eq:46}) is true for $d-1$ for some $d\ge2$. We have that
  \begin{eqnarray}
    A_d(n,t)&=&\sum_{i_1=t+1}^{n-d+1}\frac{1}{i_1}A_{d-1}(i_1,n),\\
&=&\sum_{i_1=t+1}^{n-d+1}\frac{1}{i_1}\left[\frac{1}{(d-1)!}\left(\log\frac{n}{i_1}\right)^{d-1}\pm O\left(\frac{\left(\log\frac{n}{i_1}\right)^{d-2}}{i_1}\right)\right],\\
\label{eq:47}&=&\frac{1}{(d-1)!}\sum_{i_1=t+1}^{n-d+1}\frac{1}{i_1}\left(\log\frac{n}{i_1}\right)^{d-1}\pm \sum_{i_1=t+1}^{n-d+1}O\left(\frac{\left(\log\frac{n}{i_1}\right)^{d-2}}{i^2_1}\right),
  \end{eqnarray}
where the induction hyptothesis is used in the first step. The first term in (\ref{eq:47}) can be approximated by the integral 
\begin{equation}
  \label{eq:45}
  I=\int_{x=i_1}^{n-d+2}\frac{1}{x}\left(\log\frac{n}{x}\right)^{d-1}dx,  
\end{equation}
with the error bounded as
\begin{eqnarray}
  \label{eq:49}
  \left|\sum_{i_1=t+1}^{n-d+1}\frac{1}{i_1}\left(\log\frac{n}{i_1}\right)^{d-1}-I\right|
&\le& \sum_{i_1=t+1}^{n-d+1}\left|\frac{1}{i_1}\left(\log\frac{n}{i_1}\right)^{d-1}-\frac{1}{i_1+1}\left(\log\frac{n}{i_1+1}\right)^{d-1}\right|,\\
&\le &\sum_{i_1=t+1}^{n-d+1}\left|\frac{1}{i_1}\left(\log\frac{n}{i_1}\right)^{d-1}-\frac{1}{i_1+1}\left(\log\frac{n}{i_1}\right)^{d-1}\right|,\\
\label{eq:50}&=& \sum_{i_1=t+1}^{n-d+1}\frac{\left(\log\frac{n}{i_1}\right)^{d-1}}{i_1(i_1+1)},\\
\label{eq:53}&=& O\left(\frac{\left(\log\frac{n}{t}\right)^{d-1}}{t}\right),
\end{eqnarray}
where the last step follows by taking the maximum term in (\ref{eq:50}) and multiplying by the number of terms. The second term in (\ref{eq:47}) also has the same order behavior as (\ref{eq:50}), which can, once again, be verified by taking the maximum term and multiplying by the number of terms. Now, the integral in (\ref{eq:45}) simplifies as
\begin{eqnarray}
  \label{eq:51}
  I&=&-\int_{x=i_1}^{n-d+2}\left(\log\frac{n}{x}\right)^{d-1}d\left(\log\frac{n}{x}\right),\\
\label{eq:52}&=&\frac{1}{d}\left(\log\frac{n}{i_1}\right)^{d-1}-\frac{1}{d}\left(\log\frac{n}{n-d+2}\right).
\end{eqnarray}
Using (\ref{eq:52}) and (\ref{eq:53}) in (\ref{eq:47}), the proposition is proved.
\end{proof}
Using Proposition \ref{prop:app} in (\ref{eq:44}), we get
\begin{eqnarray}
  \label{eq:54}
  \pr(S_n=d)&=&\frac{r(r-1)\cdots(r-m+2)}{n(n-1)\cdots(n-m+2)}(m-1)^dA_d(r-m+1,n-m+1),\\
&\to&\frac{r^{m-1}}{n^{m-1}}(m-1)^d\frac{1}{d!}\left(\log\frac{n}{t}\right)^d,\\
&=&\left(\frac{r}{n}\right)^{m-1}\frac{1}{d!}\left((m-1)\log_e\frac{n}{r}\right)^d,\\
\end{eqnarray}
where we have used the assumptions that $m=o(n)$ and $r$ grows linearly with $n$.

\bibliographystyle{../Work/TIFR/Research/IEEEtran}
\bibliography{../Work/TIFR/Research//IEEEabrv,../Work/TIFR/Research/Research}
\end{document}